\documentclass[letterpaper, 10 pt, conference]{ieeeconf}  

\IEEEoverridecommandlockouts                              
\usepackage{lmodern}
\usepackage[T1]{fontenc}
\usepackage{graphicx}
\usepackage{epsfig} 
\usepackage{fancyhdr}
\usepackage{mathptmx} 
\usepackage{times} 
\usepackage{amsmath} 
\usepackage{amssymb}  
\usepackage{mathrsfs}
\usepackage{url}
\usepackage{bm}
\usepackage{theorem}

\newcommand{\be}[1]{\begin{equation}\label{#1}}
\newcommand{\benon}{\begin{equation*}}  
\newcommand{\bemuln}[1]{\begin{multline}\label{#1}}
\newcommand{\bemul}{\begin{multline*}}
\newcommand{\bee}{\begin{eqnarray*}}
\newcommand{\eee}{\end{eqnarray*}}
\newcommand{\been}[1]{\begin{eqnarray}\label{#1}}
\newcommand{\eeen}{\end{eqnarray}}
\newcommand{\began}[1]{\begin{gather}\label{#1}}
\newcommand{\bega}{\begin{gather*}}
\newcommand{\bealn}[1]{\begin{align}\label{#1}}
\newcommand{\beal}{\begin{align*}}
\newcommand{\bealatn}[2]{\begin{alignat}{#1}\label{#2}}
\newcommand{\bealat}{\begin{alignat*}}
\newcommand{\bexalatn}[1]{\begin{xalignat}\label{#1}}
\newcommand{\bexalat}{\begin{xalignat*}}

{\theoremstyle{plain} \newtheorem{thm}{Theorem}[section]

\newtheorem{lemma}[thm]{Lemma}
}
{\theoremstyle{break} \theorembodyfont{\it}

 }

\def\bx{{\mathbf x}}

\def\bI{{\mathbf I}}

\def\bP{{\mathbf P}}
\def\bQ{{\mathbf Q}}

\def\bU{{\mathbf U}}

\def\bY{{\mathbf Y}}

\def\bZ{{\mathbf Z}}

\def\texitem#1{\par\smallskip\noindent\hangindent 25pt
               \hbox to 25pt {\hss #1 ~}\ignorespaces}

\newcommand{\bGamma}{\boldsymbol{\Gamma}}

\newcommand{\bLambda}{\boldsymbol{\Lambda}}

\makeatletter
\newcommand{\rmnum}[1]{\romannumeral #1}
\newcommand{\Rmnum}[1]{\expandafter\@slowromancap\romannumeral #1@}
\makeatother

\title{\LARGE \bf An Improved Composite Hypothesis Test for Markov Models with
  Applications in Network Anomaly Detection\authorrefmark{1} \thanks{*
    Research partially supported by the NSF under grants CNS-1239021 and
    IIS-1237022, by the ARO under grants W911NF-11-1-0227 and
    W911NF-12-1-0390, by the ONR under grant N00014-10-1-0952, and by the
    Cyprus Research Promotion Foundation under Grant New Infrastructure
    Project/Strategic/0308/26.}}

\author{Jing Zhang$^\dag$ and Ioannis Ch. Paschalidis$^\ddag$%
  \thanks{$\dag$ Division of Systems Engineering, Boston University,
    {\tt \small jzh@bu.edu}.}
  \thanks{$\ddag$ Department of Electrical and
  Computer Engineering and Division of Systems Engineering, Boston University,
  8 St. Mary's St., Boston, MA 02215, {\tt \small yannisp@bu.edu,
    http://sites.bu.edu/paschalidis}.}}

\begin{document}

\maketitle
\thispagestyle{empty}
\pagestyle{empty}

\begin{abstract}

  Recent work has proposed the use of a composite hypothesis Hoeffding
  test for statistical anomaly detection. Setting an appropriate
  threshold for the test given a desired false alarm probability
  involves approximating the false alarm probability. To that end, a
  large deviations asymptotic is typically used which, however, often
  results in an inaccurate setting of the threshold, especially for
  relatively small sample sizes. This, in turn, results in an anomaly
  detection test that does not control well for false alarms. In this
  paper, we develop a tighter approximation using the Central Limit
  Theorem (CLT) under Markovian assumptions. We apply our result to a
  network anomaly detection application and demonstrate its advantages
  over earlier work.

\end{abstract}


\section{Introduction}

During the last decade, the applications of statistical anomaly
detection in communication networks using the Large Deviations Theory
(LDT) \cite{dembo1998large} have been extensively explored; see, e.g.,
\cite{paschalidis2009spatio}, \cite{CDC09}, \cite{robust-anomaly-tcns}, among
others. Statistical anomaly detection involves characterizing the
``normal behavior'' of the system and identifying the time instances
corresponding to abnormal system behavior. Assuming that the network
traffic is stationary in time, \cite{paschalidis2009spatio} applies two
methods to that end. The first method, termed ``\textit{model-free},''
models traffic as an independent and identically distributed (i.i.d.)
sequence. The second method, termed ``\textit{model-based},'' models
traffic as a finite-state Markov chain. These are then extended in
\cite{robust-anomaly-tcns} to the case where the network traffic is
time-varying. Essentially, each of these methods is designed to tackle a
certain Universal Hypothesis Testing Problem (UHTP).

A UHTP aims to test the hypothesis that a given sequence of observations
is drawn from a known \textit{Probability Law (PL)} (i.e.,
\textit{probability distribution}) defined on a finite alphabet
\cite{TIT13}. It is well known that the test proposed by Hoeffding
\cite{hoeffding1965} is optimal in an error exponent sense
\cite{dembo1998large}, \cite{paschalidis2009spatio},
\cite{robust-anomaly-tcns}. When 
implementing Hoeffding's test, one must set a threshold $\eta$, which
can be estimated by using Sanov's theorem \cite{dembo1998large}, a
result within the scope of LDT. Note that such an estimate (let us
denote it by $\eta^{sv}$) is valid only in the asymptotic sense. In
practice, however, only a finite number of observations are available,
and it can be shown by experiments (e.g., using the software package
TAHTIID \cite{TAHTIID}) that $\eta^{sv}$ is typically not accurate
enough.

To improve the accuracy of the estimation for $\eta$, \cite{TIT13} (or,
see \cite{unnikrishnan2011universal}) borrows a method typically used by
statisticians; that is, deriving results based on Weak Convergence (WC)
of the test statistic to approximate the error probabilities of
Hoeffding's test. Under the i.i.d. assumption, by using the technique
combining WC results together with Taylor's series expansion, \cite{TIT13} (or,
see \cite{unnikrishnan2011universal}, \cite{wilks1938large}) gives an
alternative approximation for $\eta$ (let us denote it by
$\eta^{wc}$). Using \cite{TAHTIID}, one can verify that, in the finite
sample-size setting, $\eta^{wc}$ is typically much more accurate than
$\eta^{sv}$.

It is worth pointing out that some researchers have also tried to obtain
a more accurate approximation for $\eta$ by refining Sanov's theorem
\cite{iltis1995sharp}. However, as noted in
\cite{unnikrishnan2011universal}, such refinements of large deviation
results are typically hard to calculate.

In this work we establish $\eta^{wc}$ under a Markovian assumption,
thus, extending the work of \cite{TIT13} which considered the
i.i.d. case. We apply the proposed procedure to network anomaly
detection by embedding it into the software package SADIT \cite{SADIT}.

The rest of the paper is organized as follows. We first formulate the
problem in Sec.~\ref{sec:prob} and derive the theoretical results in
Sec.~\ref{sec:theory}. We then present the experimental results in
Sec.~\ref{sec:num} and finally provide concluding remarks in
Sec.~\ref{sec:con}.

\textbf{Notation:} By convention, all vectors are column vectors. To
save space, we write $\bx = (x_1, \ldots, x_n)$ to denote the column
vector $\bx$. We use prime to denote the transpose of a
matrix or vector.


\section{Problem Formulation} \label{sec:prob}

Let $\Xi = \left\{{\xi _i}; ~ i = 1, \ldots, N \right\}$ be a finite
alphabet. The observations $\bY = \left\{ {{Y_l};~l =
    0,1,2, \ldots } \right\}$ form a stochastic process taking values
from $\Xi$.

In the \textit{model-based} UHTP \cite{dembo1998large}, \cite{paschalidis2009spatio}, \cite{robust-anomaly-tcns}, under the null hypothesis $\cal H$, the observations  $\bY = \left\{ {{Y_l}; ~l = 0,1,2, \ldots } \right\}$ are drawn according to a Markov chain with state set $\Xi$. Let its transition matrix be $\bQ = \left( {{q_{ij}}} \right)_{i,j = 1}^N$.

Let $\textbf{1}_{\{\cdot\}}$ denote the indicator function. Define the \textit{empirical PL} by
\begin{align} {\Gamma _n}\left( {{\theta_{ij}}} \right) =
  \frac{1}{n}\sum\nolimits_{l = 1}^n {\textbf{1}_{\left\{ {{Z_l} =
          {\theta_{ij}}} \right\}}}, \label{1}
\end{align}
where ${Z_l} = \left( {{Y_{l - 1}},~{Y_l}} \right), ~l=1,\ldots,n$,
${\theta_{ij}} = \left( {{\xi _i},~{\xi _j}} \right) \in \Xi \times \Xi
$, $i,j = 1, \ldots ,N$, and denote the new alphabet $\Theta = \left\{
  {{\theta_{ij}}};~i, j = 1, \ldots, N \right\} = \left\{ {{{\tilde
        {\theta}}_k}}; ~k = 1, \ldots, N^2 \right\}$.  Note that $\Theta
= \Xi \times \Xi$ and ${{\tilde {\theta}}_1} = {\theta_{11}}$, $\ldots$,
${{\tilde {\theta}}_N} = {\theta_{1N}}$, $\ldots$, ${{\tilde
    {\theta}}_{\left( {N - 1} \right)N + 1}} = {\theta_{N1}}$, $\ldots$,
${{\tilde {\theta}}_{{N^2}}} = {\theta_{NN}}$.  Let also the set of PLs
on $\Theta$ be $\mathcal{P}(\Theta)$.

It is seen that the transformed observations $\bZ = \left\{ {{Z_l};~l =
    1,2, \ldots } \right\}$ form a Markov chain evolving on
$\Theta$. Let its transition matrix be $\bP = \left( {{p_{ij}}}
\right)_{i,j = 1}^{{N^2}}$ with stationary distribution
\[{\boldsymbol{\pi}} = \left( {{\pi _{ij}}};~i, j = 1, \ldots, N \right)
= \left( {{{\tilde {\pi}}_k}};~k = 1, \ldots, N^2 \right),\] where
${{\pi _{ij}}}$ denotes the probability of seeing $\theta_{ij}$, and
${{\tilde {\pi}}_1} = {{\pi}_{11}}$, $\ldots$, ${{\tilde {\pi}}_N} =
{{\pi}_{1N}}$, $\ldots$, ${{\tilde {\pi}}_{\left( {N - 1} \right)N + 1}}
= {{\pi}_{N1}}$, $\ldots$, ${{\tilde {\pi}}_{{N^2}}} = {{\pi}_{NN}}$.
Let $p\left( {{\theta_{ij}}\left| {{\theta_{kl}}} \right.} \right)$
denote the transition probability from $\theta_{kl}$ to
$\theta_{ij}$. Then we have
\[p\left( {{\theta_{ij}}\left| {{\theta_{kl}}} \right.} \right) =
\textbf{1}_{\left\{ {i = l} \right\}}{q_{ij}}, \quad k,l,i,j = 1,
\ldots ,N,\] which enables us to obtain $\bP$ directly from $\bQ$.

Write
${\boldsymbol{\pi} } = ({{\pi _{11}}}, \ldots, {{\pi _{1N}}}, \ldots, {{\pi _{N1}}}, \ldots, {{\pi _{NN}}})$, and ${\boldsymbol{\Gamma} _n} = ({{\Gamma _n}\left( {{\theta_{11}}} \right)}, \ldots, {{\Gamma _n}\left( {{\theta_{1N}}} \right)}, \ldots, {{\Gamma _n}\left( {{\theta_{N1}}} \right)}, \ldots, {{\Gamma _n}\left( {{\theta_{NN}}} \right)}) $.

Define the \textit{relative entropy} (i.e., \textit{divergence})
\begin{align}
  D\left( {\left. {{\boldsymbol{\Gamma} _n}} \right\|\boldsymbol{\pi} }
  \right) = \sum\limits_{i = 1}^N {\sum\limits_{j = 1}^N {{\Gamma
        _n}\left( {{\theta_{ij}}} \right)\log \frac{{{{{\Gamma _n}\left(
                {{\theta_{ij}}} \right)} \mathord{\left/ {\vphantom
                  {{{\Gamma _n}\left( {{\theta_{ij}}} \right)}
                    {\sum\nolimits_{j = 1}^N {{\Gamma _n}\left(
                          {{\theta_{ij}}} \right)} }}} \right.
              \kern-\nulldelimiterspace} {{\sum\nolimits_{j = 1}^N
                {{\Gamma _n}\left( {{\theta_{ij}}} \right)}
              }}}}}{{{{{\pi _{ij}}} \mathord{\left/ {\vphantom {{{\pi
                        _{ij}}} {{\sum\nolimits_{j = 1}^N {{\pi _{ij}}}
                      }}}} \right.  \kern-\nulldelimiterspace} {
              {\sum\nolimits_{j = 1}^N {{\pi _{ij}}} } }}}}} },
\label{2}
\end{align}
and the \textit{empirical measure}
\begin{align}
{\bU_n} = \sqrt n \left( {{\boldsymbol{\Gamma} _n} - \boldsymbol{\pi} } \right).
\label{2.0}
\end{align}

For any positive integer $n$, let ${{\cal{H}}_n}$ be the output of the
test that decides to accept or to reject the null hypothesis $\cal H$
based on the first $n$ observations in the sequence $\bZ$.  Under the
Markovian assumption, Hoeffding's test \cite{dembo1998large} is given by
\begin{align}
{{\cal{H}}_n} = {\cal{H}} \Leftrightarrow D\left( {\left. {{\boldsymbol{\Gamma} _n}} \right\|{\boldsymbol{\pi}}} \right) \leq \eta_n, \label{333}
\end{align}
where $D\left( {\left. {{\boldsymbol{\Gamma} _n}} \right\|{\boldsymbol{\pi}}} \right)$, defined in \eqref{2}, is the \textit{test statistic}, and $\eta_n$ is a \textit{threshold}. Note that, when applied in anomaly detection, Hoeffding's test will report an anomaly if and only if ${{\cal{H}}_n} \neq \cal H$, or, equivalently, $D\left( {\left. {{\boldsymbol{\Gamma} _n}} \right\|{\boldsymbol{\pi}}} \right) > \eta_n$.

Under hypothesis $\cal H$, the \textit{false alarm probability}
\cite{TIT13} of the test sequence $\left\{ {{{\cal H}_n}} \right\}$ as a
function of $n$ is defined by
\begin{align}
{\beta _n} = {P_{\cal H}}\left\{ {D\left( {\left. {{\boldsymbol{\Gamma} _n}} \right\|{\boldsymbol{\pi}}} \right) > \eta_n } \right\}, \label{334}
\end{align}
where $P_{\cal H}\{A\}$ denotes the probability of event $A$ under
hypothesis $\cal H$. Sanov's theorem \cite{dembo1998large} implies the
approximation for $\eta_n$
\begin{align} 
\eta_n^{sv}  \approx  - \frac{1}{n}\log \left( {{\beta _n}} \right). \label{335}
\end{align}
Given $\beta _n \in (0,1)$, our central goal in this paper is to
improve the accuracy of the approximation for $\eta_n$.


\section{Theoretical Results} \label{sec:theory}
We introduce the following condition:
\smallskip

\noindent \textbf{Condition (C)}.  {\em $\bZ= \left\{ {{Z_l}; ~l = 1,2,
      \ldots } \right\}$ is an aperiodic, irreducible, and positive recurrent Markov chain
  \cite{jones2004markov}, \cite{MEYN1993} evolving on $\Theta$ with stationary
  distribution $\boldsymbol{\pi}$, and with the same $\boldsymbol{\pi}$
  as its initial distribution.}  \smallskip

\textit{Remark 1:} Since $\Theta$ is a finite set, we know that $\bZ$ is uniformly ergodic \cite{jones2004markov}, \cite{roberts2004general}. Assuming the initial distribution to be 
$\boldsymbol{\pi}$ is for notational simplicity; our results apply for any feasible initial distribution. Note also that, under \textbf{\text{(C)}}, $\boldsymbol{\pi}$  must have full support over $\Theta$, meaning each entry in 
$\boldsymbol{\pi}$ is strictly positive. 

\begin{lemma} \label{le1}

Assume \textbf{\text{(C)}} holds. Then
\begin{align}
\frac{{{\pi _{ij}}}}{{\sum\nolimits_{t = 1}^N {{\pi _{it}}} }} = \frac{{{\pi _{ij}}}}{{\sum\nolimits_{t = 1}^N {{\pi _{ti}}} }} = {q_{ij}}, \quad i,j=1,\ldots,N.
\label{3}
\end{align}

\end{lemma}

\begin{proof}
Expanding the first $N$ entries of $\boldsymbol{\pi} \bP = \boldsymbol{\pi}$, 
we obtain $q_{1i}\sum\nolimits_{t = 1}^N {{\pi _{t1}}}  = {\pi _{1i}}, ~i = 1, \ldots, N$.
Summing up both sides of these equations, it follows
\begin{align}
\left( {\sum\nolimits_{i = 1}^N {{q_{1i}}} } \right) \cdot \sum\nolimits_{t = 1}^N {{\pi _{t1}}}  = \sum\nolimits_{t = 1}^N {{\pi _{1t}}}.  \label{3.2}
\end{align}
Noticing $\sum\nolimits_{i = 1}^N {{q_{1i}}}  = 1$, we see
that \eqref{3.2} implies
$\sum\nolimits_{t = 1}^N {{\pi _{t1}}}  = \sum\nolimits_{t = 1}^N {{\pi _{1t}}}$, 
which, together with $q_{11}\sum\nolimits_{t = 1}^N {{\pi _{t1}}}  = {\pi _{11}}$, yields
\[\frac{{{\pi _{11}}}}{{\sum\nolimits_{t = 1}^N {{\pi _{1t}}} }} = \frac{{{\pi _{11}}}}{{\sum\nolimits_{t = 1}^N {{\pi _{t1}}} }} = {q_{11}}.\]
Similarly, we can show \eqref{3} holds for all the other $(i,j)$'s.    
\end{proof}

\textit{Remark 2:} By \textit{Remark 1} and Lemma \ref{le1} we see that, under Condition \textbf{\text{(C)}}, all entries in $\bQ$ are strictly positive, indicating that any two states of the original chain $\bY$ are connected. This is a strict condition; yet, in engineering practice, if some $\pi_{ij}$ in \eqref{3} are equal to zero, we can always replace them with a small positive value, and then normalize the modified vector $\boldsymbol{\pi}$, thus ensuring that Condition \textbf{\text{(C)}} can be satisfied. See Sec. \ref{sec:num.A} for an illustration of this trick.

\subsection{Weak convergence of $\bU_n$}
\label{WeakConvergenceOfUn}

Let us first establish CLT results for one-dimensional empirical measures
\begin{align}
  U_n^{\scriptscriptstyle{\left( k \right)}} = \sqrt n \left( {{\Gamma
        _n}\left( {{{\tilde {\theta}}_k}} \right) - {{\tilde \pi }_k}}
  \right), \quad k = 1, \ldots ,{N^2}.
\label{2.1}
\end{align}
Fix $k \in \{1,\ldots,N^2\}$. Define
\begin{align}
{f_k}\left( Z \right) = \textbf{1}_{\left\{ {Z = {{\tilde {\theta}}_k}} \right\}} = \left\{ \begin{gathered}
    1,{\text{ if }}Z = {{\tilde {\theta}}_k}, \hfill \\
    0,{\text{ if }}Z \in \Theta\backslash \left\{ {{{\tilde {\theta}}_k}} \right\}. \hfill \\
\end{gathered}  \right.
\label{998}
\end{align}
\begin{lemma} \label{le0} Assume \textbf{\text{(C)}} holds. Then a CLT
  holds for $U_n^{\scriptscriptstyle{\left( k \right)}}$; that is,
  $U_n^{\scriptscriptstyle{\left( k \right)}}\xrightarrow[{n \to \infty
  }]{{d.}}\mathcal{N}\left( {0,~\sigma _k^2} \right)$ with $\sigma _k^2
  = \operatorname{Cov} \left( {{f_k}\left( {{Z_1}} \right),~{f_k}\left(
        {{Z_1}} \right)} \right) + 2\sum\nolimits_{m = 1}^\infty
  {\operatorname{Cov} } \left( {{f_k}\left( {{Z_1}} \right),~{f_k}\left(
        {{Z_{1 + m}}} \right)} \right) < \infty$.
\end{lemma}

\begin{proof}
  This can be established by applying \cite[Corollary
    1]{jones2004markov}. Noting $f_k(\cdot)$ is bounded and the chain $\bZ$ is
  uniformly ergodic, we see that all the conditions needed by 
    \cite[Corollary 1]{jones2004markov} are satisfied.
\end{proof}

Now we can directly extend Lemma \ref{le0}  to the multidimensional case (for
an informal argument, see, e.g., \cite{pene2005rate};
cf. \cite[Sec. \Rmnum{3}.6]{feller1957introduction} for the definition of a
Gaussian random variable with a general normal distribution). Due to its role 
in the following applications, we state the result as a theorem.
\begin{thm} \label{th1} Assume \textbf{\text{(C)}} holds. Then a
  multidimensional CLT holds for $\bU_n$; that is,
\begin{align}
{\bU_n}\xrightarrow[{n \to \infty }]{{d.}}\mathcal{N}\left( {0, ~\bLambda} \right)
\label{4}
\end{align}
with $\bLambda$ being an $N^2 \times N^2$ covariance matrix given by
\begin{align}
\bLambda =  {\bLambda _0} + \sum\nolimits_{m = 1}^\infty  {{\bLambda _m}}, 
\label{40}
\end{align}
where ${\bLambda _0}$ and ${\bLambda _m}$ are specified, respectively,
by ${\bLambda _0} = \left[ {\operatorname{Cov} \left( {{f_i}\left(
          {{Z_1}} \right),~{f_j}\left( {{Z_1}} \right)} \right)}
\right]_{i,j = 1}^{{N^2}}$ and ${\bLambda _m} = \left[
  {{{\operatorname{Cov} } }\left( {{f_i}\left( {{Z_1}}
        \right),~{f_j}\left( {{Z_{1 + m}}} \right)} \right) +
    {{\operatorname{Cov} } }\left( {{f_j}\left( {{Z_1}}
        \right),~{f_i}\left( {{Z_{1 + m}}} \right)} \right)}
\right]_{i,j = 1}^{{N^2}}$, $m = 1, 2, \ldots$.
\end{thm}

The rest of this subsection is devoted to the computation of $\bLambda$;
we will express $\bLambda = \big[{{\Lambda ^{\scriptscriptstyle{\left(
          {i,j} \right)}}}}\big]_{i,j = 1}^{{N^2}}$ in terms of the
quantities determining the probabilistic structure of the chain $\bZ$.

In particular, we have
\begin{align}
  \Lambda _0^{\scriptscriptstyle{\left( {i,j} \right)}} =
  {\text{Cov}}\big( {\textbf{1}_{\{ {{Z_1} = {{\tilde {\theta}}_i}}
      \}},~\textbf{1}_{\{ {{Z_1} = {{\tilde {\theta}}_j}} \}}} \big),
  \quad i,j = 1, \ldots ,{N^2},
\label{6}
\end{align}
and
\begin{align}
  \Lambda _m^{\scriptscriptstyle{\left( {i,j} \right)}} ~=~ & {\text{Cov}}\big( {\textbf{1}_{\{ {{Z_1} = {{\tilde {\theta}}_i}} \}},~\textbf{1}_{\{ {{Z_{1 + m}} = {{\tilde {\theta}}_j}} \}}} \big)  \nonumber \\
  &+ {\text{Cov}}\big( {\textbf{1}_{\{ {{Z_1} = {{\tilde {\theta}}_j}}
      \}},~\textbf{1}_{\{ {{Z_{1 + m}} = {{\tilde {\theta}}_i}} \}}}
  \big), \label{7}
\end{align}
$m = 1,2, \ldots ;~i,j = 1, \ldots ,{N^2}$.

We first determine \eqref{6} as follows.  By direct calculation, we
obtain $\Lambda _0^{\scriptscriptstyle{\left( {1,1} \right)}} = {{\tilde
    \pi }_1}\left( {1 - {{\tilde \pi }_1}} \right)$, $\Lambda
_0^{\scriptscriptstyle{\left( {1,2} \right)}} = - {{\tilde \pi
  }_1}{{\tilde \pi }_2}$, and so on. In general, we have $\Lambda
_0^{\scriptscriptstyle{\left( {i,j} \right)}} = \tilde{\pi}_i \left(
  \bI_{ij} - \tilde{\pi}_j \right), ~i, j = 1, \ldots ,{N^2}$, where
$\bI_{ij}$ denotes the $(i,j)$-th entry of the identity matrix.

Now we compute $\Lambda _2^{\scriptscriptstyle{\left( {i,j} \right)}}$
by \eqref{7}.  Omitting the details, again by direct calculation, we
arrive at ${\text{Cov}}\big( {\textbf{1}_{\{ {{Z_1} = {{\tilde
          {\theta}}_i}} \}},\textbf{1}_{\{ {{Z_3} = {{\tilde
          {\theta}}_j}} \}}} \big) = {{\tilde \pi
  }_i}\big({{\bP^2_{ij}} - {{\tilde \pi }_j}} \big)$, where
$\bP^2_{ij}$ is the $(i,j)$-th entry of the matrix $\bP^2$ (the square
of the transition matrix $\bP$).  Interchanging the indexes $i,j$, we
obtain ${\text{Cov}}\big( {\textbf{1}_{\{ {{Z_1} = {{\tilde
          {\theta}}_j}} \}},\textbf{1}_{\{ {{Z_3} = {{\tilde
          {\theta}}_i}} \}}} \big) = \tilde{\pi}_j \big( \bP^2_{ji} -
  \tilde{\pi}_i \big)$.  Thus, we have $\Lambda
_2^{\scriptscriptstyle{\left( {i,j} \right)}} = {{\tilde \pi }_i}\big(
  \bP_{ij}^2 - \tilde{\pi}_j \big) + {{\tilde \pi }_j}\big(
  \bP_{ji}^2 - \tilde{\pi}_i \big)$.

Computing in a similar way, in general, we derive
\[\Lambda _m^{\scriptscriptstyle{\left( {i,j} \right)}} = {{\tilde \pi }_i}\left( {\bP_{ij}^m - {{\tilde \pi }_j}} \right) + {{\tilde \pi }_j}\left( {\bP_{ji}^m - {{\tilde \pi }_i}} \right),\quad m = 1,2, \ldots, \]
where $\bP^m_{ij}$ is the $(i,j)$-th entry of the matrix $\bP^m$ (the $m$-th power of the transition matrix $\bP$).

Thus, for $i, j = 1, \ldots ,{N^2}$ we finally attain
\begin{equation}
{\Lambda ^{\scriptscriptstyle{\left( {i,j} \right)}}} = {{\tilde \pi }_i}\left( {{\bI_{ij}} - {{\tilde \pi }_j}} \right) + \sum\nolimits_{m = 1}^\infty  {\left[{{\tilde \pi }_i}\left( {\bP_{ij}^m - {{\tilde \pi }_j}} \right) + {{\tilde \pi }_j}\left( \bP_{ji}^m - \tilde{\pi}_i \right)\right]}.  \label{8}  
\end{equation}

\subsection{Taylor's series expansion of ${D}\left( {{\boldsymbol{\Gamma} _n}\left\| {{\boldsymbol{\pi}}} \right.} \right)$}

Inspired by \cite{TIT13} (or, see \cite{unnikrishnan2011universal})
wherein the i.i.d. assumption is required, we use a Taylor's series
expansion to approximate the relative entropy ${D}\left(
  {{\boldsymbol{\Gamma} _n}\left\| {{\boldsymbol{\pi}}} \right.}
\right)$.

To this end, for $\boldsymbol{\nu} \in
  \mathcal{P}\left( \Theta \right)$ let us consider 
\begin{align}
  h\left( \boldsymbol{\nu} \right) = {D}\left( {\boldsymbol{\nu} \left\|
        \boldsymbol{\pi} \right.} \right) = \sum\limits_{i = 1}^N
  {\sum\limits_{j = 1}^N {{\nu _{ij}}\log \frac{{\frac{{{\nu
                _{ij}}}}{{\sum\nolimits_{t = 1}^N {{\nu _{it}}}
            }}}}{{\frac{{{\pi _{ij}}}}{{\sum\nolimits_{t = 1}^N {{\pi
                  _{it}}} }}}}} }.
\label{90}
\end{align}

By direct calculation, we derive
\begin{align}
\frac{{\partial h\left( {\boldsymbol{\nu }} \right)}}{{\partial {\nu _{ij}}}} =& \log {\nu _{ij}} - \log \left( {\sum\nolimits_{t = 1}^N {{\nu _{it}}} } \right) - \log {\pi _{ij}}  \nonumber \\
&+ \log \left( {\sum\nolimits_{t = 1}^N {{\pi _{it}}} } \right),  \quad i,j = 1, \ldots ,N, \label{92_updated} 
\end{align}
which leads to 
\begin{align}
\nabla h\left( \boldsymbol{\pi}  \right) = 0. 
\label{93}
\end{align}

Further, from \eqref{92_updated}, we compute the Hessian ${\nabla ^2}h\left( \boldsymbol{\nu}
\right)$ by considering three cases:

(\rmnum{1}) If $k \ne i$, then $\frac{{{\partial ^2}h\left(
      \boldsymbol{\nu} \right)}}{{\partial {\nu _{ij}}\partial {\nu
      _{kl}}}} = 0$.

(\rmnum{2}) If $k=i$ and $l=j$, then
\[\frac{{{\partial ^2}h\left( \boldsymbol{\nu} \right)}}{{\partial {\nu
      _{ij}}\partial {\nu _{kl}}}} = \frac{{{\partial ^2}h\left(
      \boldsymbol{\nu} \right)}}{{\partial \nu _{ij}^2}} =
\frac{1}{{{\nu _{ij}}}} - \frac{1}{{\sum\nolimits_{j = 1}^N {{\nu
        _{ij}}} }}.\]

(\rmnum{3}) If $k=i$ and $l \ne j$, then
\[\frac{{{\partial ^2}h\left( \boldsymbol{\nu}  \right)}}{{\partial {\nu _{ij}}\partial {\nu _{kl}}}} = \frac{{{\partial ^2}h\left( \boldsymbol{\nu}  \right)}}{{\partial {\nu _{ij}}\partial {\nu _{il}}}} =  - \frac{1}{{\sum\nolimits_{j = 1}^N {{\nu _{ij}}} }}.\]
All these terms evaluated at $\boldsymbol{\pi}$ yield ${\nabla
  ^2}h\left( \boldsymbol{\pi} \right)$, which will play a crucial role
in approximating ${D}\left( {{\boldsymbol{\Gamma} _n}\left\|
      {{\boldsymbol{\pi}}} \right.} \right)$.

From the ergodicity of the chain $\bZ$, it follows ${\boldsymbol{\Gamma}
  _n}\xrightarrow[{n \to \infty }]{{a.s.}}\boldsymbol{\pi} $. It is seen
that both $\nabla h\left( \boldsymbol{\nu} \right)$ and ${\nabla
  ^2}h\left( \boldsymbol{\nu} \right)$ are continuous in a neighborhood
of $\boldsymbol{\pi}$, and we can use the second-order Taylor's
series expansion of $h\left( \boldsymbol{\nu} \right)$ centered at
$\boldsymbol{\pi}$ to approximate ${D}\left( {{\boldsymbol{\Gamma}
      _n}\left\| {{\boldsymbol{\pi}}} \right.} \right) =
h(\boldsymbol{\Gamma} _n) - h\left( \boldsymbol{\pi}  \right)$. To be specific, we have
\begin{align}
{D}\left( {{\boldsymbol{\Gamma} _n}\left\| \boldsymbol{\pi}  \right.} \right)  \approx \frac{1}{2}{\left( {{\boldsymbol{\Gamma} _n} - \boldsymbol{\pi} } \right)^{\prime}}{\nabla ^2}h\left( \boldsymbol{\pi}  \right)\left( {{\boldsymbol{\Gamma} _n} - \boldsymbol{\pi} } \right).
\label{94}
\end{align}

\subsection{Threshold approximation}

We use an empirical CDF to approximate the actual CDF of $D\left(
  {\left. {{\boldsymbol{\Gamma} _n}} \right\|{\boldsymbol{\pi}}}
\right)$. In particular, by \eqref{94} we have
\begin{align} 
{D}\left( {{\boldsymbol{\Gamma} _n}\left\| \boldsymbol{\pi}  \right.} \right) \approx \frac{1}{{2n}}{\left( {\sqrt n \left( {{\boldsymbol{\Gamma} _n} -
          \boldsymbol{\pi} } \right)} \right)^{\prime}}{\nabla
  ^2}h\left( \boldsymbol{\pi} \right)\left( {\sqrt n \left(
      {{\boldsymbol{\Gamma} _n} - \boldsymbol{\pi} } \right)} \right),
\label{336}
\end{align}
where, by Thm. \ref{th1}, $\sqrt n \left( {{\boldsymbol{\Gamma} _n} -
    \boldsymbol{\pi} } \right)\xrightarrow[{n \to \infty
}]{{d.}}\mathcal{N}\left( {0,~\bLambda } \right)$ with $\bLambda$ given
by \eqref{8}.  Thus, to derive an empirical CDF, we can generate a set
of Gaussian sample vectors independently according to $\mathcal{N}\left(
  {0, ~\bLambda } \right)$ and then plug each such sample vector into
\eqref{336} (i.e., replace $\sqrt n \left( {{\boldsymbol{\Gamma} _n} -
    \boldsymbol{\pi} } \right)$) and compute the right hand side of
\eqref{336}, thus, obtaining a set of sample scalars, as an
approximation for samples of ${D}\left( {{\boldsymbol{\Gamma} _n}\left\|
      \boldsymbol{\pi} \right.} \right)$.

Once we obtain an empirical CDF of ${D}\left( {{\boldsymbol{\Gamma}
      _n}\left\| \boldsymbol{\pi} \right.} \right)$, say, denoted
${F_{em}}\left( \cdot ; ~n \right)$, then, by \eqref{334}, we can use it
to estimate $\eta_n$ as
\begin{align} 
\eta_n^{wc}  \approx F_{em}^{ - 1}\left( {1 - {\beta _n}}; ~n \right),
\label{337}
\end{align}
where $F_{em}^{ - 1}(\cdot; ~n)$ is the inverse of $F_{em}(\cdot; ~n)$.

Clearly, to calculate the right hand side of \eqref{336}, we need the
parameter $\boldsymbol{\pi}$. In practice, we typically do not know the
actual values of the underlying chain parameters, so we need to estimate
them accordingly. This can be done by computing $\bGamma_n$ over a long
past sample path.


\section{Experimental Results} \label{sec:num}

\subsection{Numerical results for threshold approximation} \label{sec:num.A}

The experiments in this subsection are performed using the software
package TAHTMA \cite{TAHTMA}. For convenience, we use $\Theta =
\{1,2,\ldots,N^2\}$ to indicate the states, and assume $\boldsymbol{\pi}
$ (ground truth) is the initial distribution.

In the following numerical examples, we first randomly create a valid
(i.e., such that \textbf{\text{(C)}} holds) $N \times N$ transition
matrix $\bQ$, hence an $N^2 \times N^2$ transition matrix $\bP$, and
then generate $T$ samples of the chain $\bZ$, each with length $n$,
denoted ${\bZ^{\scriptscriptstyle{\left( t \right)}}} = \big\{
{Z_1^{\scriptscriptstyle{\left( t \right)}}, \ldots
  ,Z_n^{\scriptscriptstyle{\left( t \right)}}} \big\},~t = 1, \ldots
,T$. These samples will be used to derive empirical CDF's. Also, to
estimate the parameters of interest, we generate one more training
sample ${\bZ^{\scriptscriptstyle{\left( 0 \right)}}} = \big\{
{Z_1^{\scriptscriptstyle{\left( 0 \right)}}, \ldots
  ,Z_{n_0}^{\scriptscriptstyle{\left( 0 \right)}}} \big\}$, where $n_0$
is the length. The stationary distribution $\boldsymbol{\pi} $ (estimation) of the
chain is computed by taking any row of $\bP^{m_0}$, where $m_0$ is a
sufficiently large integer.

With samples ${\bZ^{\scriptscriptstyle{\left( t \right)}}} = \big\{
{Z_1^{\scriptscriptstyle{\left( t \right)}}, \ldots
  ,Z_n^{\scriptscriptstyle{\left( t \right)}}} \big\},~t = 1, \ldots
,T$, we can compute $T$ samples of the scalar random variable ${D}\left(
  {{\boldsymbol{\Gamma} _n}\left\| {{\boldsymbol{\pi}}} \right.}
\right)$, by \eqref{2}. Thus, an empirical CDF of ${D}\left(
  {{\boldsymbol{\Gamma} _n}\left\| {{\boldsymbol{\pi}}} \right.}
\right)$, denoted ${F}\left( \cdot; ~n \right)$, can be derived. We will
treat ${F}\left( \cdot; ~n \right)$ as the actual CDF of ${D}\left(
  {{\boldsymbol{\Gamma} _n}\left\| {{\boldsymbol{\pi}}} \right.}
\right)$. The threshold given by \eqref{337} with ${F}\left( \cdot; ~n
\right)$ (i.e., $F(\cdot; ~n)$ plays the role of $F_{em}(\cdot; ~n)$) is
then taken as the theoretical (actual) value $\eta_n$.

Next we derive an estimated empirical CDF of ${D}\left(
  {{\boldsymbol{\Gamma} _n}\left\| {{\boldsymbol{\pi}}} \right.}
\right)$, using Taylor's series expansion together with Thm.~\ref{th1}.

Let us first estimate the parameters of interest.  Recall
${\boldsymbol{\pi}} = \left( {{\pi _{ij}}};~i, j = 1, \ldots, N \right)
= \left( {{{\tilde {\pi}}_k}};~k = 1, \ldots, N^2 \right)$.  We estimate
$ {{{\tilde {\pi}}_k}}$ by
\begin{align} {{\hat {\tilde {\pi}} }_k} = \max \left\{
    {\frac{{\sum\nolimits_{i = 1}^{{n_0}} {\textbf{1}_{\big\{
              {Z_i^{\scriptscriptstyle{\left( 0 \right)}} = k} \big\}}}
        }}{{{n_0}}}, ~\varepsilon } \right\}, \quad k=1,\ldots,N^2,
\label{700}
\end{align}
where $\varepsilon > 0$ is a small number. The purpose of introducing
$\varepsilon$ is to avoid division by zero.  Then, we obtain an
estimated $\boldsymbol{\pi}$ as $\hat {\boldsymbol{\pi}} = \big(
{{{{{\hat{ \tilde {\pi}} }_k}}} \mathord{\left/ {\vphantom {1 2}}
    \right. \kern-\nulldelimiterspace} {{\sum\nolimits_{j = 1}^{{N^2}}
      {{{\hat{ \tilde{ \pi}} }_j}} }}}; ~k = 1, \ldots, N^2 \big)$, or
written $\hat {\boldsymbol{\pi}} = \left( {{{\hat \pi }_{ij}}}; ~i, j =
  1, \ldots, N \right)$, where ${{\hat \pi }_{11}} = {{{\hat {\tilde{
          \pi }}}_1}} \mathord{\left/ {\vphantom {1 2}}
  \right. \kern-\nulldelimiterspace} {{\hat s}}$, $\ldots$, ${{\hat \pi
  }_{1N}} = {{{{\hat {\tilde{ \pi }}}_N}}} \mathord{\left/ {\vphantom {1
        2}} \right. \kern-\nulldelimiterspace} {{\hat s}}$, $\ldots$,
${{\hat \pi }_{N1}} = {{{{\hat {\tilde{ \pi }}}_{\left( {N - 1} \right)N
        + 1}}}} \mathord{\left/ {\vphantom {1 2}}
  \right. \kern-\nulldelimiterspace} {{\hat s}}$, $\ldots$, ${{\hat \pi
  }_{NN}} = {{{{\hat {\tilde{ \pi }}}_{{N^2}}}}} \mathord{\left/
    {\vphantom {1 2}} \right. \kern-\nulldelimiterspace} {{\hat s}}$,
with $\hat s = \sum\nolimits_{j = 1}^{{N^2}} {{{\hat {\tilde{ \pi
        }}}_j}} $ being the normalization constant.

We are now ready to calculate the estimate for ${\nabla ^2}\left( \pmb{\pi}  \right)$, by evaluating ${\nabla ^2}\left( \hat{\pmb{\pi}}  \right)$.

Further, by Lemma \ref{le1}, we estimate $\bQ$ by
\begin{align} {{\hat q}_{ij}} = \frac{{{{\hat \pi
        }_{ij}}}}{{\sum\nolimits_{t = 1}^N {{{\hat \pi }_{it}}} }},
  \quad i,j = 1, \ldots ,N.
\label{703}
\end{align}
Then an estimation of $\bP$, denoted $\hat{\bP}$, can be directly derived from $\hat{\bQ}$.

Finally, by \eqref{8}, we estimate the covariance matrix $\bLambda$ using
\begin{align}
  {{\hat \Lambda }^{\scriptscriptstyle{\left( {i,j} \right)}}} ~=~& {{\hat {\tilde \pi} }_i}\left( {{{\mathbf{I}}_{ij}} - {{\hat {\tilde \pi }}_j}} \right) + \sum\nolimits_{m = 1}^{{m_0}} {} \big[{{\hat {\tilde \pi }}_i}\left( {\hat {\mathbf{P}}_{ij}^m - {{\hat {\tilde \pi} }_j}} \right) \nonumber \\
   &+ {{\hat {\tilde \pi} }_j}\left( {\hat {\mathbf{P}}_{ji}^m - {{\hat {\tilde \pi} }_i}} \right)\big], \quad i,j = 1, \ldots ,{N^2}.  
\label{704}
\end{align}
It is worth noting that, to ensure the estimated $\bLambda$ is at least
symmetric, we update $\hat{\bLambda}$ by ${(\hat{\bLambda} +
  \hat{\bLambda}^{\prime}) \mathord{\left/ {\vphantom {1 2}} \right.
    \kern-\nulldelimiterspace} 2}$. Also, to ensure $\hat \bLambda$ to
be positive semi-definite, we use the trick of $QR$ factorization
\cite{horn2012matrix}; for details, the reader is referred to
\cite{TAHTMA}.

We now generate $T$ Gaussian sample vectors according to $\mathcal{N}(
{0,~\hat{\bLambda} } )$ and then plug each of them into \eqref{336} and
compute the right hand side of \eqref{336} with ${\nabla ^2}\left( \boldsymbol{\pi}
\right)$ replaced by ${\nabla ^2}\left( \hat{\boldsymbol{\pi}} \right)$, thus,
obtaining $T$ scalar samples as approximations for samples of ${D}\left(
  {{\boldsymbol{\Gamma} _n}\left\| \boldsymbol{\pi} \right.}
\right)$. Therefore, an estimated empirical CDF of ${D}\left(
  {{\boldsymbol{\Gamma} _n}\left\| {{\boldsymbol{\pi}}} \right.}
\right)$, denoted ${F_{em}}\left( \cdot; ~n \right)$, can be derived
accordingly. $\eta_n^{wc}$ given by \eqref{337} is then the WC
approximation for the threshold.

Let the false alarm probability be $\beta_n = \beta = 0.001$. We set the
aforementioned parameters as $N = 12$, $\varepsilon = 10^{-8}$, $T =
1000$, $m_0 = 1000$, and let $n_0 = 1000N^2$. In
Fig. \ref{eta_comp_N_12}, the red line indicates the theoretical value
of the threshold $\eta_n$, the blue line shows the WC approximation
$\eta_n^{wc}$, and the green line demonstrates the estimate
$\eta_n^{sv}$ obtained by Sanov's theorem, all as a function of the
number of samples $n$. It is seen that $\eta_n^{wc}$ is much more
accurate than $\eta_n^{sv}$.

\begin{figure}[thpb]  
   \centering
   \includegraphics[scale=0.13]{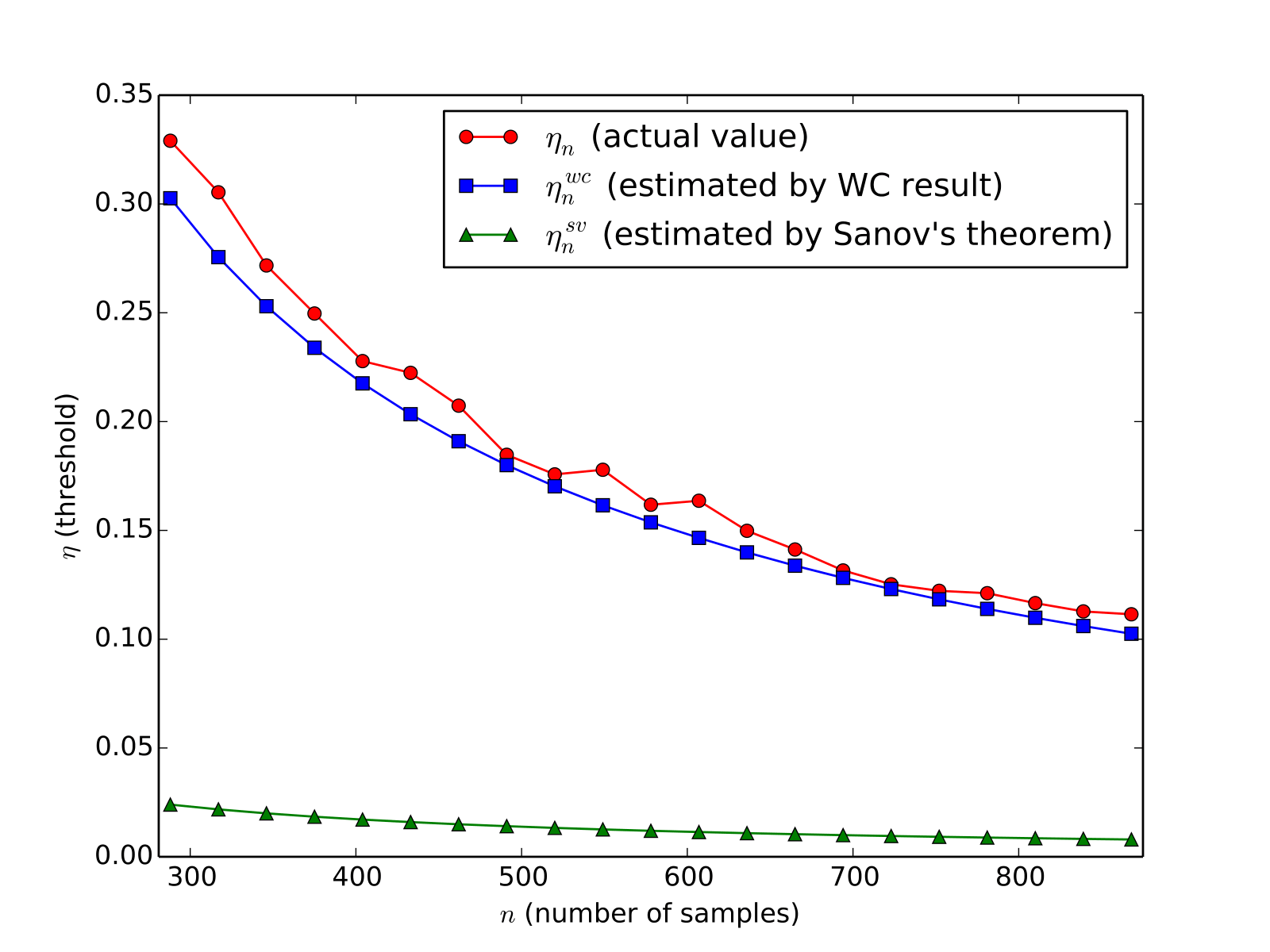}
   \caption{Threshold ($\eta$) versus number of samples ($n$); the number of states in the original Markov chain is $N = 12$.}
\label{eta_comp_N_12}
\end{figure}

\textit{Remark 3:} In Fig. \ref{eta_comp_N_12}, the red line
corresponding to the actual value $\eta_n$ is not as smooth; the reason
is that each time when varying the number of samples $n$, we regenerate
all the samples $\bZ^{\scriptscriptstyle{\left( t \right)}} = \big\{
{Z_1^{\scriptscriptstyle{\left( t \right)}}, \ldots
  ,Z_n^{\scriptscriptstyle{\left( t \right)}}} \big\}, ~t = 1, \ldots
,T$, from scratch. On the other hand, the blue line corresponding to
$\eta_n^{wc}$ looks smooth; this is because we only need to generate the
$T$ Gaussian sample vectors according to $\mathcal{N}( {0,
  ~\hat{\bLambda} } )$ once. Note that, in our experiments, most of the
running time is spent generating the samples
${\bZ^{\scriptscriptstyle{\left( t \right)}}} = \big\{
{Z_1^{\scriptscriptstyle{\left( t \right)}}, \ldots
  ,Z_n^{\scriptscriptstyle{\left( t \right)}}} \big\},~t = 1, \ldots
,T$, and then calculating $\eta_n$. In practice, when implementing the
approximation procedure proposed in this subsection, we will only need
to focus on obtaining $\eta_n^{wc}$, which is computationally
inexpensive.

\subsection{Simulation results for network anomaly detection}

We will use the term \textit{traffic} and \textit{flow} interchangeably in this subsection; they mean the same thing. The simulations are done using the software package SADIT \cite{SADIT}, which, based on the \textit{fs}-simulator \cite{fs}, can efficiently generate flow-level network traffic datasets with annotated anomalies. 

The network (see Fig. \ref{simulation_setting}) consists of an internal network involving 8 normal users (\textit{CT1-CT8}) and 1 server (\textit{SRV}) that stores some sensitive information, and 3 Internet nodes (\textit{INT1-INT3}) that visit the internal network via a gateway (\textit{GATEWAY}). 

\begin{figure}[thpb]  
   \centering
   \includegraphics[scale=0.4]{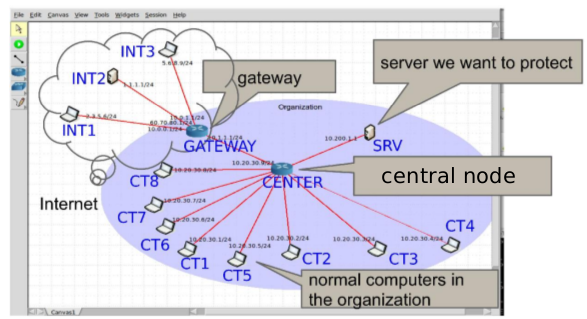}
   \caption{Simulation setting (this figure is from \cite{robust-anomaly-tcns} (or, see \cite{CDC13})).}
\label{simulation_setting}
\end{figure}

When dealing with the data, we use the features ``flow duration,'' ``flow
size,'' and ``distance to cluster center'' (this is related to a user IP
address). The data preprocessing involves three steps: \rmnum{1})
\textit{network traffic representation and flow aggregation}, \rmnum{2})
\textit{quantization}, and \rmnum{3}) \textit{window aggregation};
cf. \cite[Sec. \Rmnum{3}.A]{robust-anomaly-tcns}.

To implement Hoeffding's test, we first estimate a PL (resp., a set of PLs) for the stationary (resp., time-varying) normal traffic. For the former case, we use as the reference traffic the whole flow sequence with anomalies at some time interval that we seek to detect. For the latter case, we generate the reference traffic and the traffic with anomalies separately, sharing all of the parameters except the ones regarding the anomalies. Using the windowing technique \cite{CDC13}, \cite{robust-anomaly-tcns}, we then persistently monitor the traffic and report anomalies instantly whenever the relative entropy is higher than the threshold corresponding to the detection window.  

The key difference between the whole detection procedures for the two types of traffic is that, estimating a PL for the stationary traffic is straightforward, while, for the time-varying traffic, we need to make an effort to identify several different PLs corresponding to certain time intervals. We apply the two-steps procedure proposed in \cite{robust-anomaly-tcns}; that is, with some rough estimation of the shifting and periodic parameters of the traffic, we first generate a large set of PLs, and then refine the large family of PLs by solving  a weighted set cover problem. For further theoretical details and the implementation, the reader is referred to \cite{robust-anomaly-tcns}, \cite{SADIT}.

Note that, to deal with the time-varying traffic, \cite{robust-anomaly-tcns} proposes a generalized Hoeffding's test, which we will use for Scenario 2 in the following. However, there is no essential difference compared to the typical Hoeffding's test. Note also that, in the following experiments, we only use the \textit{model-based} method \cite{CDC13}, \cite{robust-anomaly-tcns}.

\subsubsection{Scenario 1 (Stationary Network Traffic)}

We mimic the case for anomaly caused by ``large access rate;''
cf. \cite[Sec. \Rmnum{4}.A.3]{CDC13}. The simulation time is $7000$
s. A user suspiciously increases its access rate to 10 times of its
normal value between $4000$ s and $4500$ s.  The interval between the
starting point of two consecutive time windows is taken as $50$ s, the
window-size is chosen to be $w_s = 400$ s, and the false alarm
probability is set as $\beta = 0.001$.

Set the quantization level for ``flow duration,'' ``flow size,'' and ``distance to cluster center'' to be $n_1 = 1$, $n_2 = 2$, and $n_3 = 2$, respectively. Set the number of user clusters as $k = 3$. 

\begin{figure}[thpb]  
   \centering
   \includegraphics[scale=0.13]{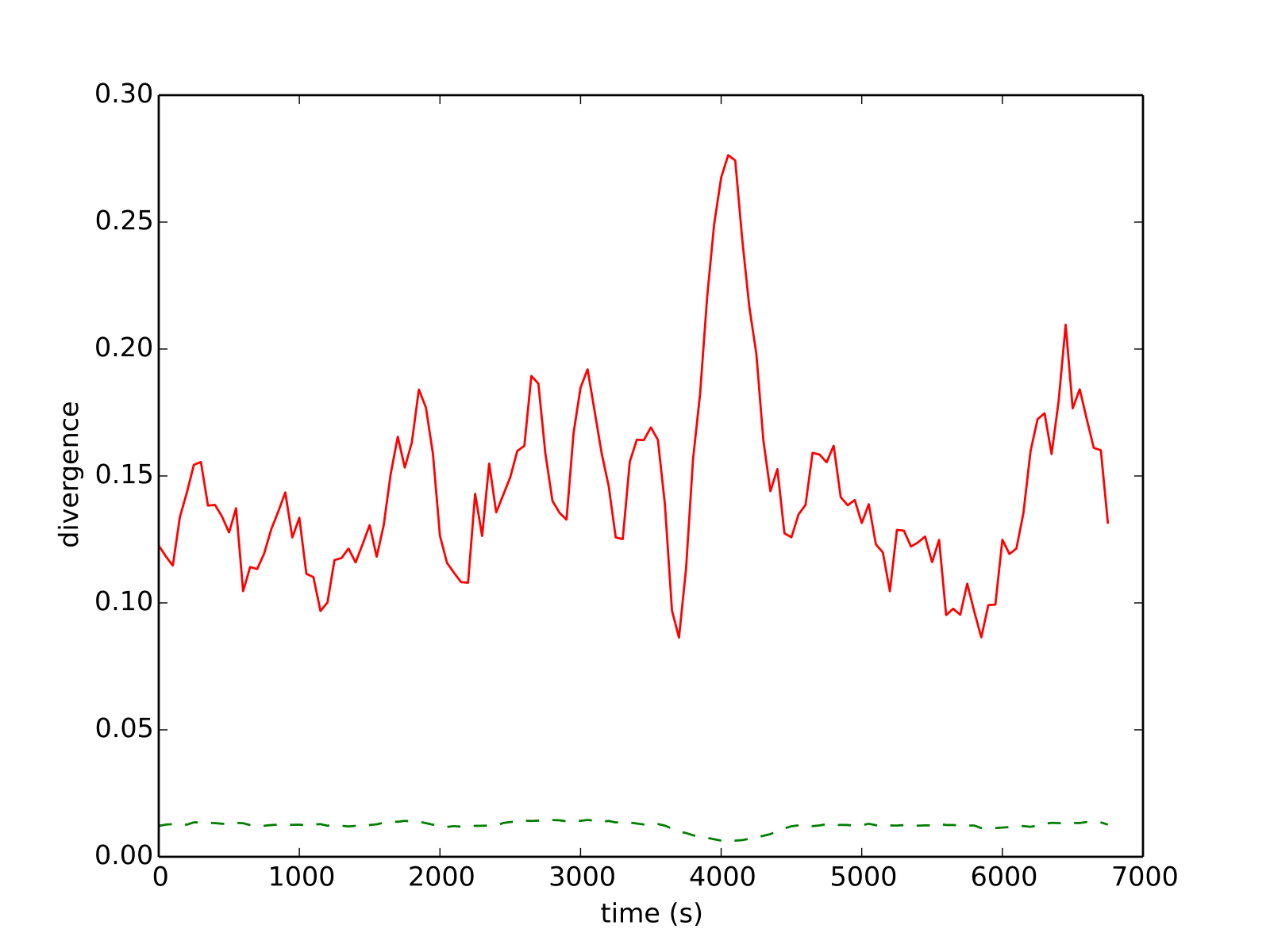}
   \caption{Detection result for Scenario 1 with $w_s = 400$ s, $n_1 = 1$, $n_2 = 2$, $n_3 = 2$, $k =3$; threshold is estimated by use of Sanov's theorem.}
\label{LargeRateSanov}
\end{figure}

\begin{figure}[thpb]  
   \centering
   \includegraphics[scale=0.13]{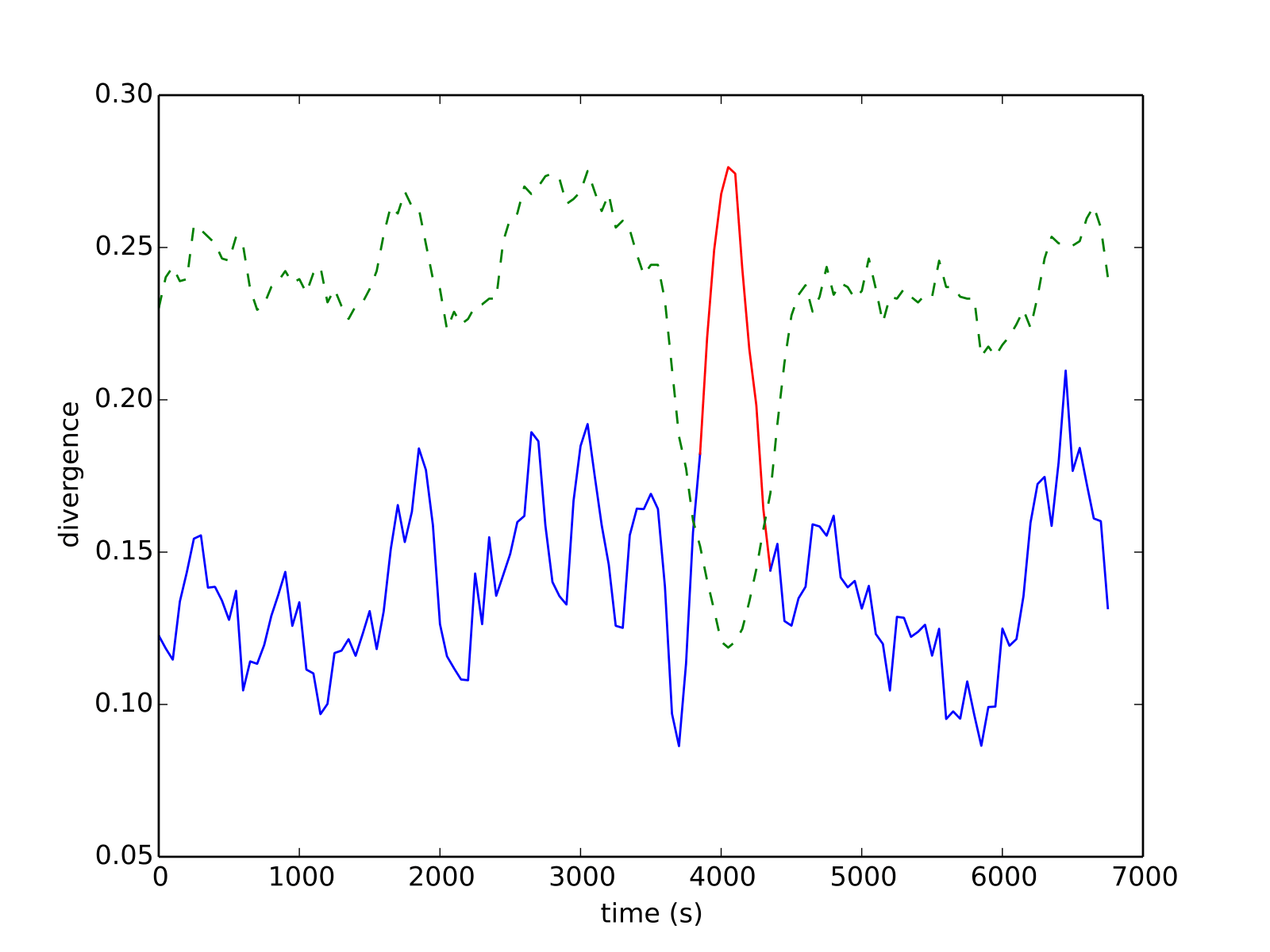}
   \caption{Detection result for Scenario 1 with $w_s = 400$ s, $n_1 = 1$, $n_2 = 2$, $n_3 = 2$, $k =3$; threshold is estimated by use of weak convergence result.}
\label{LargeRateTaylor}
\end{figure}

See Fig. \ref{LargeRateSanov} and Fig. \ref{LargeRateTaylor} for the
detection results. Both figures depict the relative entropy (divergence)
metric in \eqref{2}. The green dashed line in Fig.~\ref{LargeRateSanov}
is the threshold calculated using Sanov's theorem (i.e., $\eta_n^{sv}$,
where $n$ is the number of flows in a specific detection window). The
green dashed line in Fig. \ref{LargeRateTaylor} is the threshold
computed using the WC result derived in Sec. \Rmnum{3} (i.e.,
$\eta_n^{wc}$).  The interval during which the entropy curve is above
the threshold line (the red part) is the interval reported as abnormal.

Fig. \ref{LargeRateSanov} shows that, if using $\eta_n^{sv}$ as the
threshold, then the detection method reports an anomaly for all windows
including the ones wherein the behavior of the network traffic is
actually normal; in other words, there are too many false
alarms. Clearly, this is undesirable. Fig. \ref{LargeRateTaylor} shows
that, if, instead, we use $\eta_n^{wc}$ as the threshold, then the
detection method will not report any false alarm while successfully
identifying the true anomalies between $4000$ s and $4500$ s. The reason
for such detection results is that $\eta_n^{wc}$ is more accurate than
$\eta_n^{sv}$.

\subsubsection{Scenario 2 (Time-Varying Network Traffic)}

Consider the network with day-night pattern where the flow size follows
a log-normal distribution. We use exactly the same scenario as that in
\cite[Sec. \Rmnum{4}.B.2]{robust-anomaly-tcns}.

Applying the same procedure as in \cite[Sec. \Rmnum{3}.C]{robust-anomaly-tcns},
we first obtain $32$ rough PLs. Then, using the heuristic PL refinement
algorithm given in \cite[Sec. \Rmnum{3}.D]{robust-anomaly-tcns} equipped with
the up-bound nominal cross-entropy parameter $\lambda = 0.027565$, which
is determined by the WC approximation, we finally obtain $6$ PLs (see
Fig. \ref{time-varying-pl}). These PLs are active during morning,
afternoon, evening, mid-night, dawn, and the transition time around
sunrise, respectively.  In the following detection procedure, the key
difference between our method and the one used in \cite{robust-anomaly-tcns} is that
we no longer manually set the threshold universally as a constant;
instead, the threshold $\eta_n^{wc}$ for each detection window is
automatically calculated by use of the WC approximation. We set the
quantization level for ``flow duration,'' ``flow size,'' and ``distance
to cluster center'' to be $n_1 = 1$, $n_2 = 4$, and $n_3 = 1$,
respectively, and set the number of user clusters to $k = 1$. The
interval between the starting point of two consecutive time windows is
chosen as $1000$ s, the window-size is taken as $w_s = 1000$ s, and the
false alarm probability is set as $\beta = 0.001$.

Noting that the anomaly is simulated beginning at 59 h and lasting for
80 minutes \cite{robust-anomaly-tcns}, we see from Fig. \ref{time-varying-detect}
that the anomaly is successfully detected, without any false alarms.

\begin{figure}[thpb]  
   \centering
   \includegraphics[scale=0.13]{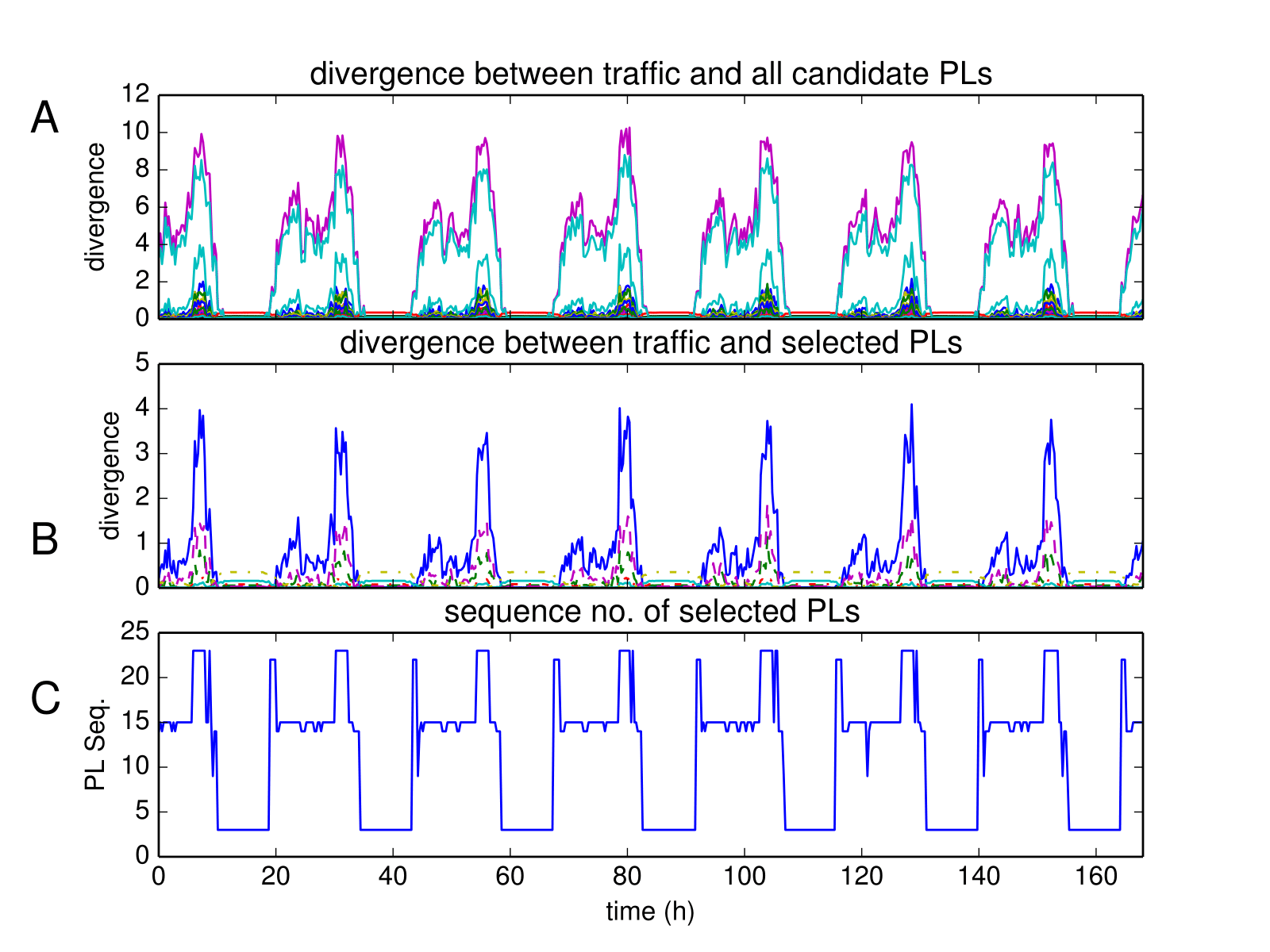}
   \caption{PLs identification for Scenario 2 with $w_s = 1000$ s, $n_1 = 1$, $n_2 = 4$, $n_3 = 1$, $k =1$.}
\label{time-varying-pl}
\end{figure}

\begin{figure}[thpb]  
   \centering
   \includegraphics[scale=0.13]{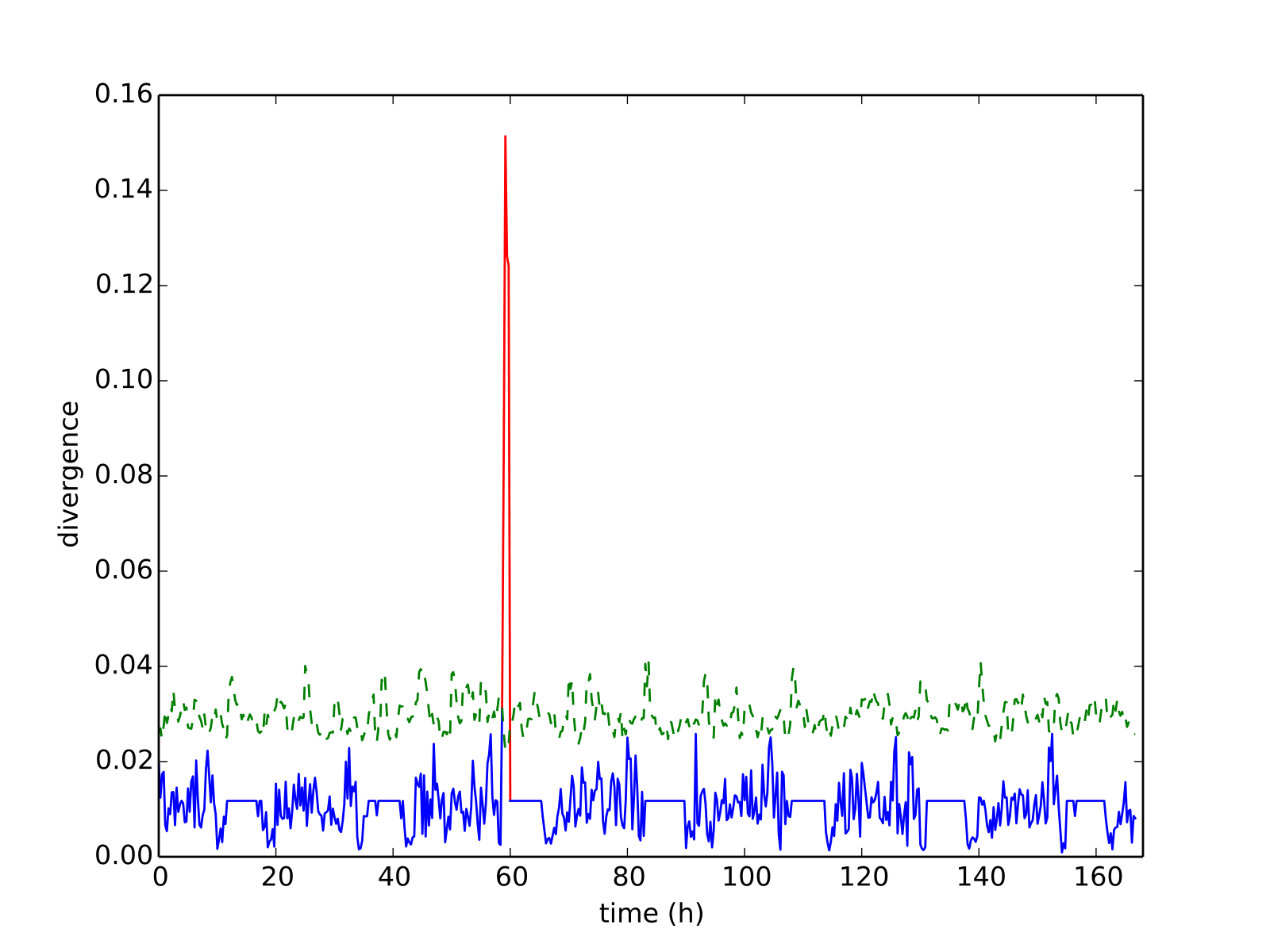}
   \caption{Detection result for Scenario 2 with $w_s = 1000$ s, $n_1 = 1$, $n_2 = 4$, $n_3 = 1$, $k =1$.}
\label{time-varying-detect}
\end{figure}

\section{Conclusions} \label{sec:con}

We establish a weak convergence result for the relative entropy in
Hoeffding's test under a Markovian assumption, which enables us to obtain
a more accurate approximation (compared to the existing approximation
given by Sanov's theorem) for the threshold needed by the test. We
validate the accuracy of such approximation through numerical
experiments and simulations for network anomaly detection.


\section*{Acknowledgements}

The authors would like to thank Dr. Jing Wang for his contributions and
help in developing the software package SADIT \cite{SADIT}.

\bibliographystyle{IEEEtran}

\bibliography{bib1}

\addtolength{\textheight}{-3cm}   

\end{document}